\documentclass[10pt, conference, letterpaper]{IEEEtran}

\usepackage[monochrome]{xcolor}
\usepackage{cite}

\usepackage{graphicx}

\usepackage{mathrsfs}
\usepackage[mathscr]{euscript}
\usepackage{enumitem}
\usepackage{amsmath,amssymb}
\usepackage{dsfont}
\usepackage{mathtools}

\usepackage{algorithm}
\usepackage{algpseudocode} 

\usepackage{array}

\usepackage{subfigure}
\usepackage{hyperref}

\usepackage{comment} 
\usepackage{cuted}

\usepackage{color}
\usepackage[english]{babel}
\usepackage[utf8]{inputenc}
\usepackage{tabularx}
\usepackage{bm}
\usepackage{bbold}

\usepackage{tikz}
\usetikzlibrary{graphs}

\newtheorem{theorem}{\bf Theorem}

\def\QED{~\rule[-1pt]{5pt}{5pt}\par\medskip}
\newenvironment{proof}{{\bf Proof: \ }}{ \hfill \QED}

\newcommand{\dU}{u}
\newcommand{\Age}{A}
\newcommand{\Gen}{X}
\newcommand{\Tr}{\tau}
\newcommand{\Tol}{\epsilon}
\newcommand{\cm}{\mu}
\newcommand{\fal}{f_{a}^{l}}
\newcommand{\yal}{y_{a}^{l}}
\newcommand{\spi}{\pi}

\newcommand{\dsp}{\displaystyle}
\newcommand{\ati}[1]{{\color{purple} #1}}

\begin{document}
%
\title{
 Age-Optimal Multi-Channel-Scheduling under Energy and Tolerance Constraints}


\author{\IEEEauthorblockN{Xujin Zhou, Irem Koprulu, Atilla Eryilmaz}\\
\IEEEauthorblockA{\textit{Electrical and Computer Engineering} \\
\textit{The Ohio State University}\\
Columbus, US \\
\{zhou.2400@osu.edu, irem.koprulu@gmail.com, eryilmaz.2@osu.edu\}}}

\maketitle

\begin{abstract}
We study the optimal scheduling problem where $n$ source nodes attempt to transmit updates over $L$ shared wireless on/off fading channels to optimize their age performance under energy and age-violation tolerance constraints. Specifically, we provide a generic formulation of age-optimization in the form of a constrained Markov Decision Processes (CMDP), and obtain the optimal scheduler as the solution of an associated Linear Programming problem. We investigate the characteristics of the optimal single-user multi-channel scheduler for the important special cases of average-age and violation-rate minimization. This leads to several key insights on the nature of the optimal allocation of the limited energy, {\color{blue}where a usual threshold-based policy does not apply and will be useful in guiding scheduler designers}. We then investigate the stability region of the optimal scheduler for the multi-user case. We also develop an online scheduler using Lyapunov-drift-minimization methods that do not require the knowledge of channel statistics. Our numerical studies compare the stability region of our online scheduler to the optimal scheduler to reveal that it performs closely with unknown channel statistics.  
\end{abstract}


\IEEEpeerreviewmaketitle

\section{Introduction}
In recent years, the Internet of Things (IoT) has become one of the most important frameworks of the next-generation wireless networks, whereby a large number of mobile devices need to be supported over an ultra-wide frequency spectrum (see, for example, \cite{federal2016fcc}). In particular, for many real-time IoT applications, it is necessary for the devices to send \textit{fresh} updates over the shared spectrum. To measure the freshness of data, the concept of \textit{Age of Information (AoI)} has been introduced over the last decade (see, for example, \cite{kaul2011minimizing,kaul2012real,kosta2017age}), which is defined concisely as the elapsed time since the generation time of the last received status update. Since the introduction of the AoI metric, numerous related studies emerged in various networking scenarios, including wireless random access networks (e.g., \cite{chen2020age,Zhou2022Efficient}), content distribution networks (e.g., \cite{abolhassani2022fresh,liu2019proactive}), scheduling (e.g., \cite{tang2019scheduling,sombabu2018age,li2020age,jhunjhunwala2018age,han2020joint}), queuing networks (e.g., \cite{hu2021status,pappas2019delay}), and vehicular networks (e.g., \cite{chen2020minimizing}). 

Recently, other AoI related metrics have been developed in order to address more generalized or different forms of ageing, such as: non-linear AoI (e.g., \cite{kosta2017age,sun2019sampling}), peak AoI (e.g., \cite{emara2020spatiotemporal}), time-since-last-service (e.g., \cite{li2017emulating}), age upon decisions (e.g., \cite{dong2019age}), to name a few. Among them, the metric, called the \textit{age-violation-rate} (see \cite{Li2021Soft,pappas2019delay,song2021analysis}) is of particular interest for real-time IoT services that have hard age-deadline constraints and a limited tolerance to violating this deadline (see \cite{devassy2018delay,hu2020optimal} for further motivation of this metric). 

In view of its significance for next generation IoT networks, in this paper, we study the general optimal multi-channel scheduling problem to optimize varying forms of age performances under energy and age-violation tolerance constraints. 
Our contributions can be listed as:
\begin{itemize}
    \item  We provide a generic formulation of age-optimization as a Constrained Markov Decision Problem (CMDP) (see \cite{altman1999constrained,dolgov2005stationary,satija2020constrained}) and obtain the age-optimal multi-channel scheduler as the solution of an associated Linear Programming problem, first for the single-source (in Section~\ref{sec:single}) and then for general the multi-source (in Section~\ref{sec:multi}) scenarios.
    \item For the single-source multi-channel scenario, we investigate the characteristics of the optimal schedulers under energy constraints for two age metrics that are important for IoT applications: (i) \textit{average-age} minimization; and (ii) \textit{age-violation-rate} minimization, {\color{blue}a non-convex/concave metric} (in Section~\ref{single:insight}). Our investigations reveal various insights on {\color{blue}different energy allocation} structures, as well as the common monotonicity properties of the optimal schedulers for minimizing these two metrics, {\color{blue}which is useful for guiding scheduler designers.}  
    \item For the multi-source age optimal scheduling problem, we also study the feasibility region of the average-age-optimal scheduler under age-violation-rate tolerance constraints to contrast its results with those of related earlier works that are developed for the single-channel multi-user scenario (see Section~\ref{sec:U2opt} and Section~\ref{sec:compare}). 
    \item Moreover, we develop (in Section~\ref{sec:online}) an \textit{online} scheduler using Lyapunov-drift-minimization methods (e.g., \cite{neely2010stochastic}) that does not require the knowledge of channel statistics, and compare its performance to the optimal and earlier designs to reveal how much the knowledge of channel statistics affects the feasibility region (see  Section~\ref{sec:compare}).
\end{itemize}

{\color{blue}Our work relates to, but also differs from several other related works in this domain. Many early works (e.g., \cite{tang2019scheduling,jhunjhunwala2018age,hsu2018age}) aim to minimize AoI under power constraints but with the assumption of reliable channels as opposed to the fading channels that we consider.
More recent works (e.g., \cite{sombabu2018age,tripathi2017age}) aim to minimize AoI-related costs based on max-age matching, while other works (e.g., \cite{zou2021minimizing,hsu2018age}) proposed AoI minimization schedulers based on \emph{Whittle Index} approach. However, to the best of our knowledge, prior works predominantly assume that one source can choose at most one channel, which is an important factor in proving the \textit{Whittle Indexability} of the corresponding problems they solve. In contrast, one of the key features our setting is the possibility of each user to transmit over multiple channels as enabled by new wireless technologies. Furthermore, most of the above mentioned works have average or peak AoI as the objective function, while we consider more general age-based objective functions, which for example allows the objective function to be a non-convex metric such as the age-violation-rate.
In this multi-channel setting with general objectives, we observe (cf. Section~\ref{single:insight}) that the optimal solution can in fact possess \textit{non-monotone} characteristics, which make the Whittle Indexability approach infeasible in general.}
The work in \cite{Li2021Soft} has considered the multi-source single-channel scheduling problem under tolerance constraints, which is a special case of our setting. We would like to note that this interesting work \cite{Li2021Soft} has been a primary motivation for our current work in exploring a different approach based on the CMDP framework that guarantees optimality and applies to more general multi-channel scenarios with additional energy constraints. There are also works (e.g., \cite{elgabli2019reinforcement,li2020learning}) that focus on learning-based approaches which can be considered as complementary to the focus of this work.

\section{System Model}

We consider the operation of a discrete-time wireless access system, whereby $n$ source nodes share $L$ 
on/off fading wireless channels to update their ageing status at a  receiver (such as a base station) under energy and violation tolerance constraints (see Figure~\ref{fig:model}). 

\begin{figure}[h]
\begin{center}
\includegraphics[height=4.5cm]{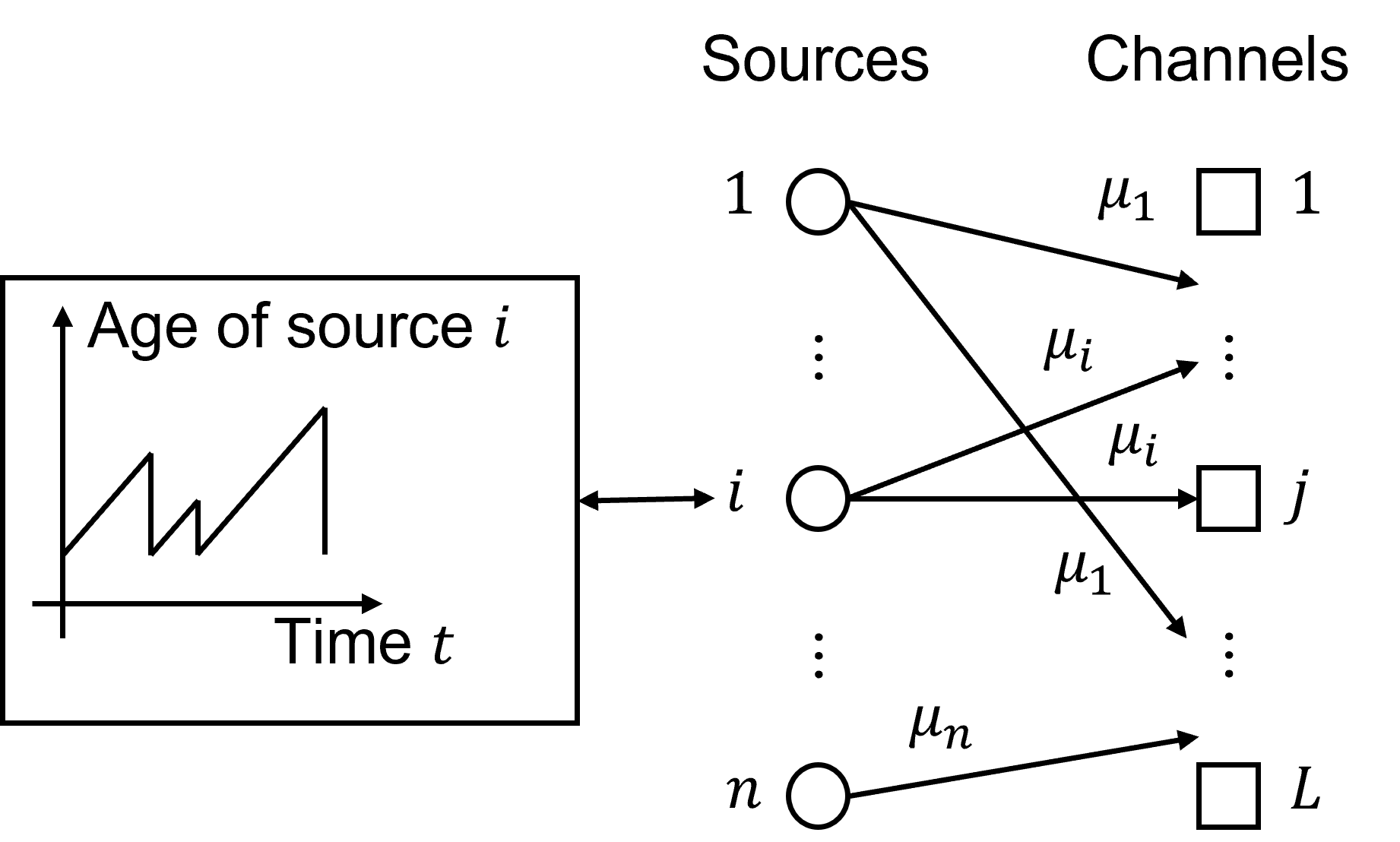}
\caption{$n$ sources share $L$ on-off fading channels to update their status to a receiver under energy and tolerance constraints in order to keep their age levels low.}
\label{fig:model}
\end{center}
\end{figure}

Our goal is to develop generic solution strategies to find optimal schedulers that can optimize diverse age-based metrics while meeting certain requirements on energy consumption and tolerance levels. We describe the key terminology and the essential system dynamics in the rest of this section. Then, in the following sections we formulate and solve classes of age-optimization problems for single and multi-source cases, subsequently.

\noindent \textbf{Scheduling policy and age-violation-tolerance:} We assume that each source node $i\in \{1, \cdots,n\}$ refreshes its status and creates a new packet at the beginning of every time slot $t\in \{1,2,3,\cdots$\}. Source nodes attempt to transmit their freshest packet to the receiver, for example a base station (BS), whenever they get a chance to transmit. Every time the BS successfully receives a new status from source node $i$, it saves the current status and discards all previous packets received from that node. As such, the BS keeps only one packet from each source node, namely the freshest one. We use $\Gen_i[t]$ to denote the generation time of the packet stored at the BS from source $i$ at time $t$. We define the \textit{age} $\Age_i[t]$ of source node $i$ at time $t$ as the time that has elapsed since the generation of its last received packet\footnote{This metric is also referred to as \textit{Age-of-Information (AoI)} and \textit{Time-Since-Last-Service (TSLS)} in different contexts. In the rest of the paper, we will refer to it as AoI or simple as age, interchangeably.}: $\Age_i[t] \triangleq t-\Gen_i[t]$.  We use\footnote{We will consistently use bold symbols to represent vectors.}  $\bm{\Age}[t] \triangleq (\Age_1[t], \cdots,\Age_n[t])$ to denote the ages of all sources at time slot $t.$

At the beginning of each time slot, the centralized scheduler decides which channels each of the source nodes will use to transmit to the base station based on the ages $\bm{\Age}[t]$ of all source nodes. Let $\dU_i(\bm{\Age}[t])$ be the number of channels source node $i$ uses to transmit at time $t$. Each transmission attempt can resolve in success or failure which we will describe below as part of the channel success model. If the base station successfully receives the packet from source $i$ at time $t$, then its age at time $t+1$ will reset to $1,$ otherwise its age will increase by one, i.e.,
\begin{equation*}
    \Age_i[t+1]= \begin{cases}1, & \hspace{-0.5in}\text {if transmission of source $i$ succeeds} \\ \Age_i[t]+1, & \text {otherwise.}\end{cases}
\end{equation*}
We allow each source $i$ to have a desired age threshold/deadline $\Tr_i$. The information of source $i$ is up-to-date if its age is less than or equal to this threshold $\Tr_i$. Otherwise, we speak of an age violation in that slot. In particular, we define the \textit{age-violation-rate} of source $i$ as the long-term average fraction of time slots when the source's age $\Age_i[t]$ exceeds its threshold $\Tr_i$, i.e., $\dsp \lim_{T\rightarrow\infty} \frac{1}{T} \dsp\sum_{t=1}^{T} \mathbb{1}\left\{\Age_i[t]> \Tr_i\right\}$. 
We use $\Tol_i\in[0,1]$ to indicate the tolerance of source $i$ that measures the maximum allowed age-violation-rate for its updates. ($\Tol_i=1$ indicates that there is no violation rate constraint, and $\Tol_i=0$ indicates that we do not allow any deadline violation.) When the age violation rate is no greater than the tolerance rate, the age violation tolerance constraint is satisfied.

\noindent \textbf{Channel success model and energy constraints:} 
The $n$ source nodes share $L$ wireless on/off fading channels, each of which can accommodate at most one packet transmission. However, even when there is a single transmission over a channel, a successful transmission is not guaranteed. 
In particular, source node $i$ has a channel success probability of $\cm_i$ when transmitting over each of its assigned channels\footnote{All our development can be generalized to the case when the success probability between source $i$ and channel $j$ is allowed to be different as $\cm_{ij}.$ However, this is omitted here as it increases the complexity of the exposition without adding to the substance.}. 

We call the update of source $i$ in a slot to be a success if any one of its transmissions over its assigned channels is successful. Since the channel is a collision channel, for an optimal scheduler we always have $\dsp\sum_{i=1}^n\dU_i(\bm{\Age}[t])\leq L.$ Once the value of $\dU_i(\bm{A}[t])$ is decided for all $i$, the scheduler will assign different channels to different sources, so that no two sources transmit over the same channel. Also, note that under the described channel success model, the probability for the BS to successfully receive an update from source node $i$ when the node uses $l$ channels is $1-(1-\cm_i)^l.$   

We assume that each transmission over a channel comes with an energy cost of $1$ unit\footnote{This can also be generalized to non-uniform energy costs over different channels, but omitted to avoid cumbersome notation.}. We require that the aggregate time-average energy cost for source $i$ is not greater than a given constraint $b_i$ channels per slot, i.e., we require $$\dsp \lim_{T \rightarrow \infty} \frac{1}{T} \dsp\sum_{t=1}^{T} \dU_i\left(\bm{\Age}[t]\right) \leq b_i, \quad b_i\in\mathbb{R^+}.$$ 

It is obvious that transmitting over more channels will increase the success probability of a source, but increase  energy consumption. We are interested in finding the number of channels that when allocated to sources optimize the desired age performance given the current age state, as well as energy and and tolerance constraints discussed above. 
In the next section, we attack this problem within the constrained Markov Decision Process (MDP) framework first for a single user, and then extend our approach to cover the multi-user setting. 

\section{Age-Optimal Multi-Channel Scheduling for a Single User}
\label{sec:single}

In this section, we first consider the single-user age-optimal multi-channel scheduling problem. This not only allows us to simplify the notation by omitting the subscripts, but also is of particular interest for the next generation ultra-wideband wireless communication technologies that are expected to support low-delay access over multiple fading channels.  We formulate a general age-optimal optimization problem which can be used in different scenarios in Section~\ref{single:form} and following the analysis of the performance in Section~\ref{single:ana}. To that end, in Section~\ref{single:insight}, we study the characterization and insights of the optimal schedulers for two important special cases of minimizing the \textit{average-age} and the \textit{age-violation-rate}, {\color{blue}which will be useful in guiding scheduler designers.}

\subsection{Problem formulation}
\label{single:form}
The problem of minimizing time-averaged age-based objectives under average energy and tolerance constraints can be generally formulated as the following constrained Markov decision problem \cite{altman1999constrained}: 
\begin{eqnarray}
&\min\limits_{\dU(\Age)} & \dsp \lim _{T \rightarrow \infty} \frac{1}{T} \sum_{t=1}^{T} \mathbb{E}\left[\omega_0(\Age[t])\right] \label{eqn:opt:single}\\
& \text { s.t }: & \dsp \lim _{T \rightarrow \infty} \frac{1}{T} \sum_{t=1}^{T} \mathbb{E}\left[\dU\left(\Age[t]\right)\right] \leq b, \label{eqn:energy}\\
& & \dsp \lim _{T \rightarrow \infty} \frac{1}{T} \sum_{t=1}^{T} \mathbb{E}\left[\omega_k\left(\Age[t]\right)\right] \leq c_k, \> k=1,\cdots,K, \nonumber \\
& & \dU(\Age[t]) \in \{0,1,\cdots,L\}.  \nonumber
\end{eqnarray}

The optimization is performed over Markovian policies described by a function $u(\cdot)$ that maps age levels to number of channels. It is known that such Markovian policies are sufficient for optimal operation\cite{altman1999constrained}. 

The first constraint on the time-averaged $u(\cdot)$ captures the average energy constraint discussed in the system model. The functions  $\omega_k(\cdot)$ serve as general functions that map the current state $A[t]$ to a value that measures the cost of that age with respect to various measures\footnote{We note that the problem can also solved with the same approach (but heavier notation) by more generally defining $\omega_k(\Age[t],\dU(\Age[t]))$ to be functions of both the age and the action.} By setting different mappings for the weight function $\omega_0(A[t])$, the objective can be changed into different commonly used age-related objectives: letting $\omega_0(a)=-\mathbb{1}\{a=1\}$ transform the objective to maximizing the average throughput; letting $\omega_0(a)=a$ makes the objective  minimize the average AoI; letting $\omega_0(a)=\mathbb{1}\{a\geq d\}$ make the objective minimize the average age-violation rate. {\color{blue} Note that this allows the objective function to be a non-convex/concave function.}
\subsection{Performance analysis}
\label{single:ana}
Next, we will analyze the generic constrained optimization problem under energy constraint by showing that the problem is equivalent to a Linear Programming (LP) problem and thus describe the optimal policy.

\begin{theorem}\label{thm:single_user}
The solution of the generic age-optimization problem \eqref{eqn:opt:single} can be obtained by solving the following linear programming problem:
\begin{equation*}
\begin{array}{ll}
\min\limits_{\yal} & \dsp \sum_{a=1}^{D}\sum_{l=0}^{L} y_{a}^{l}\omega_0(a) \\
\text {s.t:} & \dsp \sum_{a=1}^{D} \sum_{l=0}^{L} \yal \cdot l \leq b, \\
& \dsp\sum_{a=1}^{D}\sum_{l=0}^{L} y_{a}^{l}\omega_k(a)\leq c_k, \> k=1,\cdots,K,\\
& 0\leq \yal \leq 1 \quad\forall 1\leq a\leq D,0\leq l\leq L, \\
& \dsp\sum_{a=1}^{D} \sum_{l=0}^{L} \yal=1,\\ 
& \mathbf{Q} \bm{y}=\bm{0} ,\\
\end{array}
\end{equation*}
where $\bm{y}$ is a column vector of size $D L$ with $\bm{y}=(y_1^1,\cdots,y_1^L,\cdots,y_D^1,\cdots,y_D^L)^T$ as its components; $D$ is an upper bound on the age state in the system which can be set sufficiently large so that the probability of reaching $D$ is vanishing.\footnote{In practice, moderate level of $D$ is enough so that the dimension of LP won't be large. Also, when there is only age violation related objective and constraints, it's enough to set $D=d+1$. See~\ref{single:insight} and~\ref{sec:U2opt} for references.} $\mathbf{Q}\bm{y}=\bm{0}$ is the matrix representation of the following (global balance) equations:
\begin{equation*} 
\begin{array}{ll}
     &\sum\limits_{l=0}^{L} y_{a+1}^{l}-\sum\limits_{l=0}^{L} \yal (1-\cm)^{l}=0  \quad \forall a=1,\cdots,D-2,\\ 
     &\sum\limits_{l=0}^{L}\left(1-(1-\cm)^{l}\right) y_{D}^{l}-\sum\limits_{l=0}^{L} y_{D-1}^{l}(1-\cm)^{l}=0,\\
     &-\sum\limits_{l=0}^{L} y_{1}^{l}(1-\cm)^{l}+\sum\limits_{a=2}^{D} \sum\limits_{l=0}^{L} \yal \left(1-(1-\cm)^{l}\right)=0.
\end{array}
\end{equation*}
If this LP is feasible, and $\bm{y}$ is an optimal solution, then the optimal policy $\dU^*(a)$ is a probabilistic policy, whereby the probability $\fal$ of choosing $l$ channels when the age is at state $a$ equals:
\begin{equation}
    \fal= \begin{cases}\frac{\yal}{\dsp\sum_{l=0}^L\yal}, & \text {if} \dsp\sum_{l=0}^L\yal \neq 0 \\ \frac{1}{L}, & \text {if} \dsp\sum_{l}\yal =0\end{cases}
    \label{eqn:ytof}
\end{equation}
for $l=0,1,\cdots,L$ and $a=1,2,\cdots,D.$
\end{theorem}
\begin{proof}
As shown in \cite{altman1999constrained}, it is enough for us to optimize over the Markovian policies for Problem~\ref{eqn:opt:single}. Since the process is not affected by a shift in time, we can define the probabilistic scheduling policy where $\fal$ denotes the probability of choosing $l$ channels when the AoI of single source is at state $a$. The normalization constraint of the probabilistic scheduling policy requires $\dsp\sum_{l=0}^{L} \fal=1$ and $\fal \geqslant 0$ for all $a$. 

Notice that the system state can be fully characterized by a one-dimensional Markov chain with age $A[t]$ as state. Given the current state information $A[t]$, the system state at the next time slot $A[t+1]$ depends only on the current state $A[t]$ (with no dependence on earlier states) and the current action $\dU[t].$ In addition, the objective and constraints only depend on the current state and action. So an equivalent MDP problem can be formulated. Let $\lambda_{a_1}^{a_2}$ denote the transition probability from state $a_1$ to $a_2$, and define $\bar{\mu}\triangleq1-\mu$ as the probability of channel failure. Then based on the channel success model,
\begin{equation}
\label{eqn:single:tran}
    \lambda_{a_1}^{a_2}=\begin{cases} \dsp\sum_{l=1}^{L}f_{a_1}^l\bar{\mu}^l, & 1\leq a_1\leq D-1, a_2=a_1+1\\ 
    \dsp\sum_{l=1}^{L}f_a^l(1-\bar{\mu}^l), & a_1=1,\cdots,D, a_2=1\\
    \dsp\sum_{l=1}^{L}f_D^l(1-\bar{\mu}^l), & a_1=D, a_2=D\\
    0, & \text {otherwise.}
    \end{cases}
\end{equation}
Since there are finitely many states, there exists a stationary distribution $\spi(a)$ for every $a.$ Let $\mathscr{C}$ be the set of all recurrent states, then $\mathscr{C}$ is irreducible and closed, thus $\mathscr{C}$ is positive recurrent. When $a\in \mathscr{C}$ the stationary distribution $\spi(a)$ is equal to the long term average $\dsp \lim _{T \rightarrow \infty} \frac{1}{T} \sum_{t=1}^{T} \mathbb{1}\{A[t]=a\}$ independent of the starting point. When state $a\notin \mathscr{C}$, then both the stationary distribution and the long term average are equal to zero.
So the optimization problem is equivalent to the following constraint MDP problem: 

\begin{eqnarray}
\min\limits_{\fal} & \dsp\sum_{a=1}^{D} \spi(a)\omega_0(a) \nonumber\\
\text {s.t:} & \dsp\sum_{a=1}^{D} \sum_{l=0}^{L} \spi(a)\fal l \leq b \nonumber\\
&\dsp\sum_{a=1}^{D} \spi(a)\omega_k(a)\leq c_k,\> k=1,\cdots,K \label{eqn:sin:tol}\\
& \dsp\sum_{l=0}^{L} \fal=1, \fal \geqslant 0 \quad\forall a\leq D,l\leq L \\ \label{eqn:sin:f}
& H \cdot \Pi=\Pi, \quad \bm{1} \cdot \Pi=1 \label{eqn:sin:norm}
\end{eqnarray}
where $\Pi=[\spi(1),\cdots,\spi(D)]^T$ is the stationary distribution of the Markov Chain and $H$ is the $D \times D$ transition matrix with $h_{ij}=\lambda_{j}^{i}$.
Let us define $\yal=\spi(a)\fal,$ then $\spi(a)=\dsp\sum_{l=0}^L \yal$ for $a\leq D$. Then the constraint~\ref{eqn:sin:tol} becomes:
$$\dsp\sum_{a=1}^{D}\dsp\sum_{l=0}^L \yal\omega_k(a)\leq c_k, \> k=1,\cdots,K.$$
The normalization constraint in Equation~\ref{eqn:sin:norm} requires $\dsp\sum_{a=1}^{D} \sum_{l=0}^{L} \yal=1$. Substituting $\yal$ into the CMDP problem and after simplifying, we establish the equivalency of the Linear Programming problem. 
After obtaining the solution $\bm{y}$, we let $\fal=\yal/\spi(a)$ for $\spi(a)\neq 0$.States $a$ with $\spi(a)=0,$ are transient states, and the actions at these states do not affect the average results. For those states we adopt a simple policy as in Equation~\ref{eqn:ytof}, then the constraint~\ref{eqn:sin:f} is also satisfied.
\end{proof}

\subsection{Characterization and Insights on Age-Optimal Schedulers}
\label{single:insight}
Our general framework encompasses a wide range of objectives and constraints for different choices of $\omega_k(\cdot)$ functions using different age and age-violation metrics. 
In this section, we focus on two important problems that can be expressed within our framework: \emph{average age minimization} and \emph{age-violation-rate minimization}. 
This effort will enable us to characterize their optimal schedulers and gain insights into their nature.  

\noindent \textbf{Optimal scheduler minimizing average age:}
When we set $\omega_0(a)=a$ in \eqref{eqn:opt:single}, the objective of the optimization problem becomes to minimize the average age 
$$\dsp \lim _{T \rightarrow \infty} \frac{1}{T} \sum_{t=1}^{T} \mathbb{E}\{\Age[t]\} = \sum_{a=1}^{D} a \>\spi(a).$$ 
For this problem formulation, we retain the energy constraint $\dsp \lim _{T \rightarrow \infty} \frac{1}{T} \sum_{t=1}^{T} \mathbb{E}\left[\dU\left(\Age[t]\right)\right] \leq b$; but do not need additional age constraints. Hence, $\omega_k(a)=0$ and $c_k=0$, for all $k$ and $a$. 

Figure~\ref{fig:minAoIU1} depicts the average number of activated channels of the average-age optimal scheduler as a function of the age states under different channel success probabilities $\cm$ for the energy constraint $b=2$. We will further discuss these results at the end of this section in comparison with the next scheduler of interest.  

\begin{figure}[t!]
\begin{center}
\includegraphics[height=5cm]{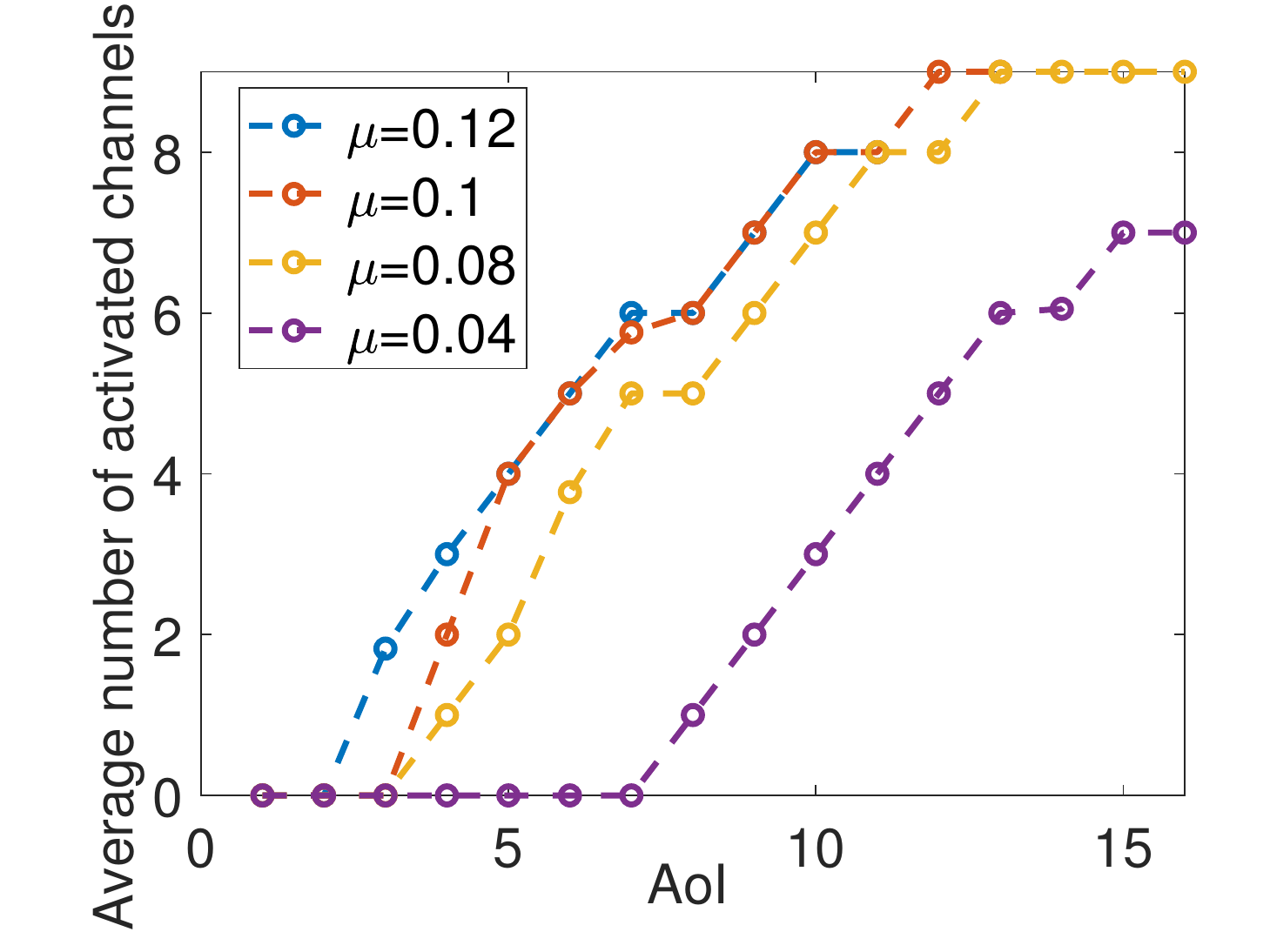}
\caption{Optimal number of channels to choose to minimize average AoI when $b=2$.}
\label{fig:minAoIU1} 
\end{center}
\end{figure}

\noindent \textbf{Optimal scheduler minimizing age-violation-rate:}
Setting $\omega_0(a)=\mathbb{1}\{a> \Tr\}$ n \eqref{eqn:opt:single}, the objective becomes minimizing the average age-violation-rate 
$$\dsp \lim _{T \rightarrow \infty} \frac{1}{T} \sum_{t=1}^{T} \mathbb{E}\{\mathbb{1}\{\Age[t]> \Tr\}\} = \dsp\sum_{a=\Tr+1}^{D} \spi(a).$$ 
As before, we keep the energy constraint, but do not need additional age constraints. Hence,  $\omega_k(a)=0$ and $c_k=0$, for all $k$ and $a$.  

With this, the problem becomes minimizing the age-violation-rate under an energy constraint. Unlike in the previous problem, our goal is not to minimize the average age but to avoid age-violation events. In this scenario, we can view all the states with $a>\Tr$ as state $\Tr+1,$ so it's enough to set $D=\Tr+1.$

Figure~\ref{fig:softAoIU1} depicts the average number of activated channels of the violation-rate optimal scheduler as a function of the age states under different channel success probabilities $\cm$ for age threshold $\Tr=8$ and the same energy constraint $b=2$. Next, we compare the optimal policies of these two schedulers and discuss the insights that can be gained from their study.

\begin{figure}[t!]
\begin{center}
\includegraphics[height=5cm]{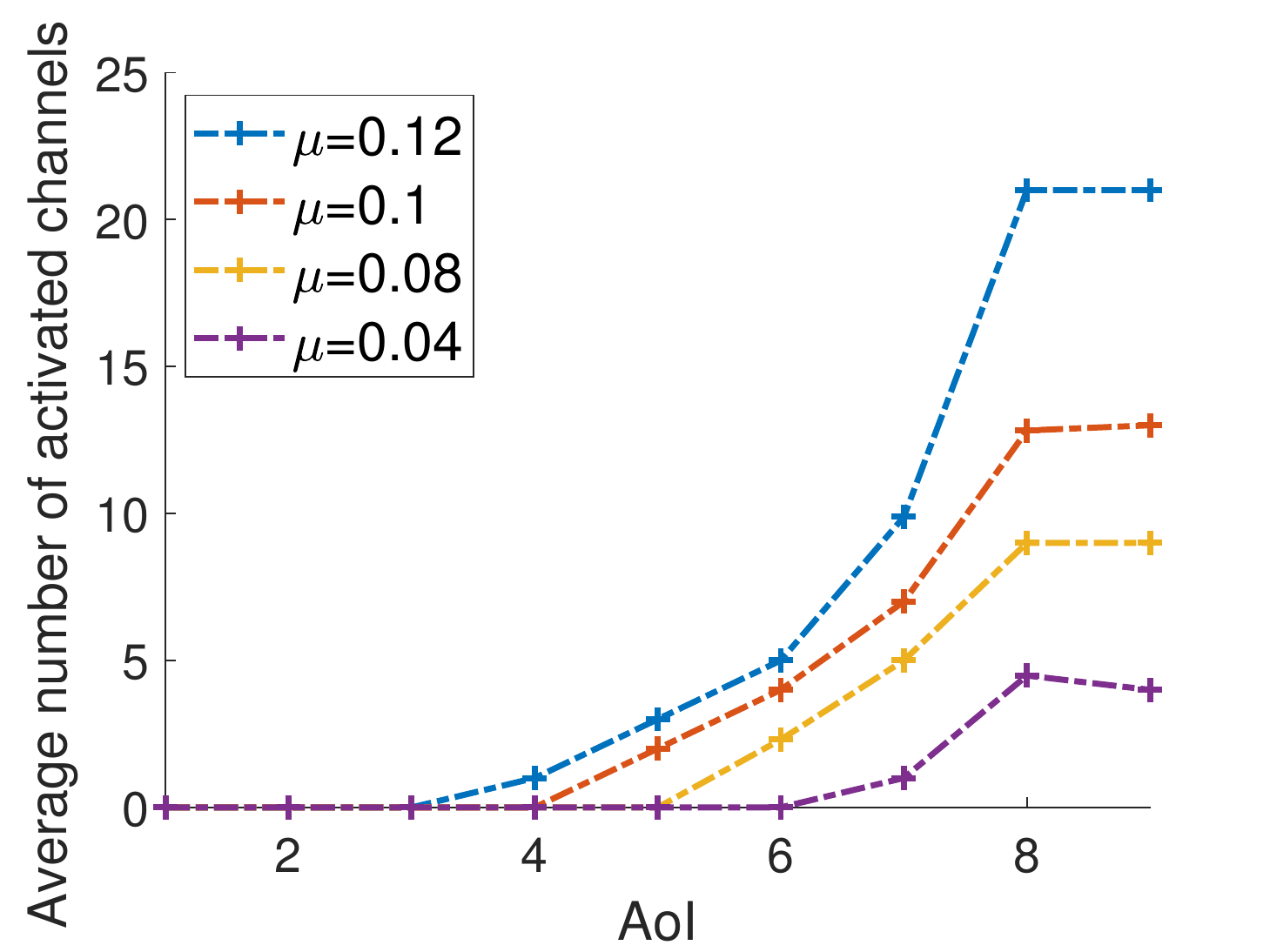}
\caption{Optimal number of channels to choose to minimize AoI violation rate when $b=2$ and $\Tr=8$.}
\label{fig:softAoIU1}
\end{center}
\end{figure}

\noindent \textbf{Insights on the two optimal schedulers:} We start by noting the similarities of the optimal policy under both scenarios:
\begin{enumerate}[label=(\roman*)]
\item Each optimal policy is a probabilistic combination of at most two deterministic policies, {\color{blue}which matches the result that the number of randomization is at most the number of constraints, as shown in \cite{altman1999constrained}.}
\item For each scenario, as the channel success probability increases, the corresponding optimal policy starts transmitting at lower age levels, and also tends to choose more channels at the same age level. This is a somewhat counter-intuitive characteristic that indicates that the optimal policy should be more active and active earlier when the channels are more reliable. 
\item {\color{blue}The optimal policy in each scenario is idle when AoI is relatively small.} This is meaningful once we observe that, when the age is relatively small, a successful transmission will not benefit the objective as much as when the age is large. Hence, the optimal scheduler saves energy for larger age states. 
\end{enumerate}

However, we also notice differences between the two sets of schedulers:
\begin{enumerate}[label=(\roman*)]
\item {\color{blue}The optimal policy in the average age minimization problem has an activation function $u^*(\cdot)$ that is monotone non-decreasing with increasing age state. On the other hand, the monotonicity does not hold in the age violation rate minimization problem. This difference comes from the non-convex nature of the \ati{the age violation rate} function in the latter case. In \cite{altman1999constrained} and many related works (e.g., \cite{tang2019scheduling,tang2020minimizing}), the authors exploit the monotone structure and threshold \ati{nature} of the optimal scheduling policy \ati{for} solving the CMDP, revealing insights as well as simplifying the algorithm by using the convexity or concavity of the objective functions. However, \ati{in our general treatment, the objective functions, such as} age violation rate, \ati{are not necessarily} convex or concave, which prevents us from using the same approach. Hence, to obtain the optimal policy, we use the generally applicable LP method despite the higher computational complexity that it may require \ati{in order to develop insights about the optimal solution}.}
\item In the average age minimization problem, the number of activated channels of the optimal policy experiences a sub-linear/concave like increase with respect to ages after the age level that the number of activated channels starts to be above zero. In contrast, the age violation rate minimizing schedulers experience a super-linear/convex like increasing with respect to age until the deadline level $\Tr.$ 
This difference can be interpreted as follows: in the age violation rate minimization problem, the penalty happens only when the age is beyond the age deadline, and hence the optimal scheduler will be more aggressive as the threshold level is approached from below. In contrast, for the average age minimization problem, the number of activated channels increases more gradually to balance the tradeoff between consuming energy unnecessarily at very low age levels and waiting too long to consume the available energy, which yields an indefinitely increasing cost. 
\end{enumerate}

These insights on the structure of the allocation functions of the optimal schedulers can guide designers in restricting their search to classes of functions with sufficiently flexible but also tractable forms whenever the solution through the LP strategy is not possible due to lack of prior statistical information as well as computational resources.

\begin{figure}
\begin{center}
\includegraphics[height=5cm]{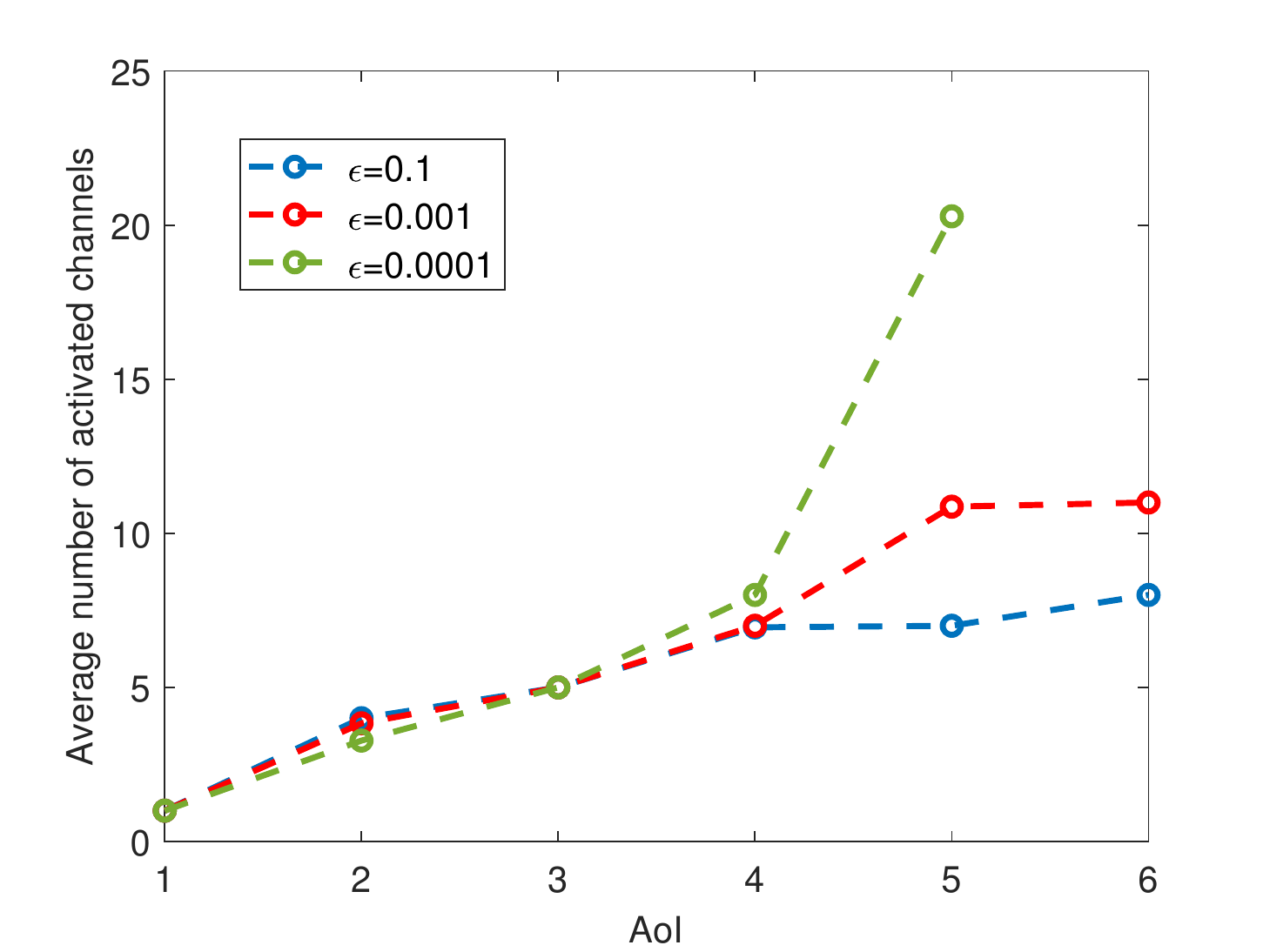}
\caption{Optimal number of channels to choose to minimize average age under violation rate constraint when $\Tr=5,b=3, \mu=0.2$}
\label{fig:transition}
\end{center}
\end{figure}

To demonstrate how the age violation rate constraint effects the shape of the scheduler more clearly, in Figure~\ref{fig:transition} we set the objective function to be $\omega_0(a)=a$, the energy constraint to be $b=3$, and the channel success probability to be $\cm=0.2$. 
In addition, we set $\omega_1(a)=\mathbb{1}\{a> \Tr\}$, where the age deadline $\Tr=5.$ We set $c_1=\Tol$ and show how the number of activated channels changes over age states under different $\Tol$ levels. By adding and tightening the tolerance constraint, we can see the transition from concave (or sublinear) to convex (or superlinear) form. As such, the optimal scheduler becomes more aggressive when the age increases. This reveals a trade-off between the average age and the age-violation-rate, namely that reducing the age violation rate calls for an increasingly more aggressive allocation function.

\section{Age-Optimal Multi-Channel Scheduling for Multiple Users}
\label{sec:multi}
In this section, we extend our framework to the general \emph{multi-user} multi-channel age-optimal scheduling problem. As before, this formulation allows us to cover a range of scenarios depending on the choice for objective function and constraints. To that end, we \ati{investigate} the feasibility and stability region of the optimal policy \ati{along with alternatives from related literature} associated with multi-user settings.

\subsection{Problem Formulation}
The formulation of the optimization problem for the multi-user case is similar to single user case \eqref{eqn:opt:single}:
\begin{eqnarray}
& \min\limits_{\bm{\dU}(\bm{\Age})} & \dsp \lim _{T \rightarrow \infty} \frac{1}{T} \sum_{t=1}^{T} \mathbb{E}\left[\omega_0(\bm{\Age}[t])\right] \label{eqn:opt:multi}\\
& \text { s.t }: & \dsp \lim _{T \rightarrow \infty} \frac{1}{T} \sum_{t=1}^{T} \mathbb{E}\left[\dU_i\left(\bm{\Age}[t]\right)\right] \leq b_i, \> i=1,\cdots,n,\nonumber \\
& & \dsp \lim _{T \rightarrow \infty} \frac{1}{T} \sum_{t=1}^{T} \mathbb{E}\left[\omega_k\left(\bm{\Age}[t]\right)\right] \leq c_k, \> k= 1,\cdots,K,\nonumber \\
& & \dU_i(\bm{\Age}[t])\in \{0,1,\cdots,L\}, \> i=1,\cdots,n,\nonumber \\
& & \dsp\sum_{i=1}^n \dU_i(\bm{\Age}[t])\leq L \nonumber 
\end{eqnarray}
where $=(\dU_1(\bm{A}),\cdots,\dU_n(\bm{A}))$ denotes the scheduling policy at state $\bm{A}$ with $\dU_i(\bm{A})$ as the number of channels allocated to source $i$. The weight functions $\omega_k(\cdot), k=0,1,\cdots,K,$ map the age states to cost values that capture age-related objectives and constraints. Source nodes can have heterogeneous energy constraints $b_i$, which means node $i$ can transmit over at most $b_i$ channels per slot on average.

\subsection{Performance analysis}
Next, we establish the equivalence of the multi-user problem formulation to a linear programming (LP) problem, as we did for the single user case in Section~\ref{single:ana}. To enable a more compact notation, we will use $\bm{a}\triangleq (a_1,a_2,\cdots,a_n)$ and $\bm{l}\triangleq (l_1,l_2,\cdots,l_n)$ to denote values of $\bm{\Age}[t]$ and $\bm{\dU}(\bm{A})$, respectively. 
We further define sets $\mathscr{A}\triangleq\{1,\cdots,D\}^n$, $\mathscr{L}\triangleq\{1,\cdots,L\}^n$, and $\mathscr{L}_1\triangleq\{\bm{l}:l_\Sigma\leq L\}$ where $l_\Sigma\triangleq\dsp\sum_{i=1}^n l_i$.

\begin{theorem}
The solution of the multi-user age-optimization problem \eqref{eqn:opt:multi} can be obtained by solving the following linear programming problem:
\begin{align}
\min\limits_{y_{\bm{a}}^{\bm{l}}} & \sum\limits_{\bm{a}\in\mathscr{A}}\sum\limits_{\bm{l}\in\mathscr{L}_1}y_{\bm{a}}^{\bm{l}}\omega_0(\bm{a}) \nonumber\\
\text { s.t: } & \dsp\sum_{\bm{a}\in\mathscr{A}} \sum_{\bm{l}\in\mathscr{L}_1} y_{\bm{a}}^{\bm{l}} l_i \leq b_i, i=1,2,\cdots,n\nonumber\\
& 0\leq y_{\bm{a}}^{\bm{l}} \leq 1 \quad \forall \bm{l}\in\mathscr{L},\bm{a}\in\mathscr{A}\nonumber\\
& y_{\bm{a}}^{\bm{l}}=0 \quad \forall \bm{l}\in\mathscr{L}/\mathscr{L}_1 \nonumber\\
& \dsp\sum_{\bm{a}\in\mathscr{A}}\sum_{\bm{l}\in\mathscr{L}_1}y_{\bm{a}}^{\bm{l}}=1\nonumber\\
&\dsp\sum_{\bm{a}\in\mathscr{A}} \sum_{\bm{l}\in\mathscr{L}_1}y_{\bm{a}}^{\bm{l}}\omega_k(\bm{a}) \leq c_k, k=1,\cdots,K\label{eqn:multiOmegak}\\
& \mathbf{Q} \bm{y}=\bm{0} \nonumber
\end{align}
where $\bm{y}$ is a column vector with $y_{\bm{a}}^{\bm{l}}$ as components and $\mathbf{Q}$ represents the transition matrix associated with the age dynamics, exactly in the same form as in the single-user case (cf. Theorem~\ref{thm:single_user}). 


If this LP is feasible and $\bm{y}$ is an optimal solution, then the optimal policy $\dU_i^*(\bm{a})$ is a probabilistic policy, whereby the probability $f_{\bm{a}}^{\bm{l}}$ of choosing $\bm{l}$ channels for source nodes $i=1,\cdots,n$ when the AoI is at state $\bm{a}$ equals:
\begin{equation*}
    f_{\bm{a}}^{\bm{l}}= \begin{cases}\frac{y_{\bm{a}}^{\bm{l}}}{\dsp\sum_{\bm{l}\in\mathscr{L}}y_{\bm{a}}^{\bm{l}}}, & \text {if} \dsp\sum_{\bm{l}\in\mathscr{L}}y_{\bm{a}}^{\bm{l}} \neq 0 \\ \frac{1}{|\mathscr{L}|}, & \text {if} \dsp\sum_{\bm{l}\in\mathscr{L}}y_{\bm{a}}^{\bm{l}} =0\end{cases}
\end{equation*}
for $\bm{l}\in\mathscr{L}$ and $\bm{a}\in \mathscr{A}.$ 
\end{theorem}
\begin{proof}
We will  use $f_{\bm{a}}^{\bm{l}}$ to denote the probability of choosing $\bm{l}=(l_1,\cdots,l_n)$ channels for source nodes $(1,\cdots,n$) when the AoI is at state $\bm{a}$. Thus $\sum\limits_{\bm{l}\in\mathscr{L}}f_{\bm{a}}^{\bm{l}}=1$, and $f_{\bm{a}}^{\bm{l}}\geq 0$ for all $\bm{a}$. 
Similarly as in Theorem~\ref{thm:single_user}, the constraint MDP problem with $n-$dimensional Markov Chains for multi-user scheduling can be generally formulated as:
\begin{eqnarray}
&\min & \dsp\sum_{\bm{a}}\spi(\bm{a})\omega_0(\bm{a}) \nonumber\\
&\text { s.t: } & \dsp\sum_{\bm{a}} \sum_{\bm{l}} \spi(\bm{a}) f_{\bm{a}}^{\bm{l}} l_i \leq b_i,i=1,2,\cdots,n \label{eqn:multi:energy}\\
& & f_{\bm{a}}^{\bm{l}}=0 \quad \forall \bm{l}\in\mathscr{L}/\mathscr{L}_1\nonumber\\
& &\dsp\sum_{\bm{a}}\spi(\bm{a})\omega_k(\bm{a})\leq c_k \quad k=1,\cdots,K\nonumber\\
& & H \cdot \Pi=\Pi,\quad \bm{1} \cdot \Pi=1, \label{eqn:multi:station}
\end{eqnarray}
where the indices range over  $\bm{a}\in\mathscr{A}$ and $\bm{l}\in\mathscr{L}$;  
$\spi(\bm{a})$ is the stationary distribution of state $\bm{a}$; and $\omega_k(\bm{a}), k=0,1,\cdots,K,$ are age related objective and cost functions. The constraints~\ref{eqn:multi:energy} bound the average energy of nodes $i$ by $b_i$ for $i=1,\cdots,n$. In the constraint~\ref{eqn:multi:station}, $\Pi$ is a $D^n\times 1$ stationary distribution vector with $\spi(\bm{a}), \bm{a}\in\mathscr{A}$ as entries.\footnote{The existence of the stationary distribution follows by the same proof as in Theorem~\ref{thm:single_user}.} $H$ represents the $D^n\times D^n$ transaction matrix with $h_{i,j}$ equals the probability of transaction from the $j^{th}$ state in $\Pi$ to the $i^{th}$ state in $\Pi,$ which can be detailed by using the age evolution and channel success probability equations similarly as in Equation~\ref{eqn:single:tran}.
Similarly, we will define $$y_{\bm{a}}^{\bm{l}}\triangleq y_{a_1,a_2,\cdots,a_n}^{l_1,l_2,\cdots,l_n}=\spi(\bm{a})f_{\bm{a}}^{\bm{l}}.$$
By changing the value of the weight functions, we can get different AoI related metrics, but all are linear with respect to $y_{\bm{a}}^{\bm{l}}$. Then,$$\spi(\bm{a})=\sum\limits_{\bm{l}}y_{\bm{a}}^{\bm{l}},$$ and the normalization constraint requires: $$\sum\limits_{\bm{a}} \sum\limits_{\bm{l}}y_{\bm{a}}^{\bm{l}}=1.$$
Substituting $y_{\bm{a}}^{\bm{l}}$ into the CMDP problem, we obtain the equivalence of the LP problem.
\end{proof}


\subsection{Characterization and insights on multi-user scheduling problem with violation tolerance Constraints}
\label{sec:U2opt}
Since there is no closed form solution to the general age-optimal problem, we will study the multi-user single-channel scheduling feasibility problem with age-violation tolerance constraint as a common setting to investigate its performance and characteristics. 

In particular, we will compare the stability region of the optimal scheduler with a previously developed algorithm that was developed for the special case of multi-user single-channel setting \cite{Li2021Soft}. To that end, we set $L=1$ and $b_i>1$. Thus, all the energy constraints will be inactive, and we can focus on the tolerance constraint, as in \cite{Li2021Soft}. Since we are only interested in feasibility, we set $\omega_0(\bm{a})=1$ for all $\bm{a}$. To express the age-violation rate constraints we define the weight functions 
$$\omega_k(\bm{a})=
\begin{cases}
0, & \text {if } a_k\leq \Tr_k \\
1, & \text {if } a_k\geq \Tr_k+1,
\end{cases}$$
and set $c_k=\Tol_k$ for $k=1,2,\cdots,K=n,$ to represent the heterogeneous age-violation tolerance level for the $k^{th}$ source. 

Then the constraint $\dsp\sum_{\bm{a}}\spi(\bm{a})\omega_k(\bm{a})\leq c_k$ becomes 
$$\spi_k(\Tr_k+1)\leq \Tol_k \quad \forall k=1,\cdots,K=n,$$
where $\spi_k(\Tr_k+1)$ denotes the total probability (under the stationary distribution) that source $k$ violates its age threshold $\Tr_k$. 
Since 
$$\spi_k(\Tr_k+1) = \sum\limits_{j_1,...,j_{k-1},j_{k+1},...,j_n} \spi(j_1,...j_{k-1},\Tr_k+1,j_{k+1}...,j_n),$$ the constraint~(\ref{eqn:multiOmegak}) in the linear programming problem becomes $$\sum\limits_{j_1,...,j_{k-1},j_{k+1},...,j_n}\sum\limits_{\bm{l}} y_{j_1,...,j_{k-1},\Tr_k+1,j_{k+1},...,j_n}^{\bm{l}} \leq \Tol_k.$$

For the sake of easy visualization, we study the case with $n=2$ users. In this case, the LP problem is formulated as:
\begin{equation*}
    \begin{array}{ll}
\min & 1 \\
\text { s.t: } & 0\leq y_{a_1,a_2}^{l_1,l_2}\leq 1 \quad \forall l_1,l_2=0,1\\
&  y_{a_1,a_2}^{l_1,l_2}=0 \quad \forall l_1+l_2>1 \\
& \sum\limits_{j}\sum\limits_{l_1,l_2} y_{\Tr_1+1,j}^{l_1,l_2}\leq \Tol_1\\
& \sum\limits_{j}\sum\limits_{l_1,l_2} y_{j,\Tr_2+1}^{l_1,l_2}\leq \Tol_2
\end{array}
\end{equation*}
The numerical results can be seen in Figures\ref{fig:U2Region_sym} and \ref{fig:U2Region_asm} for different parameters where the upper right area of the solid blue line is the stability region of the optimal scheduler. These typical examples reveal the non-negligible gap between the performance of the optimal scheduler and the previously proposed design, even for a small two user setting. 

This motivates the search for new algorithms that can perform closer to the optimal scheduler, even when the channel statistics are unknown a priori. This is performed in the next section along with further discussion about these numerical results after we discuss our online scheduling algorithm.

Before we proceed, we note even the above numerical results are for two-user single-channel scheduling problem under tolerance constraints for visualization purposes, our methods apply to the more general multi-user multi-channel scheduling problem under violation tolerance and energy constraints. {\color{blue} Although the computational complexity may be relatively high for the LP solution compared to other solutions that exploit the special structure of particular problems, as we mentioned above, due to the non-convexity and non-concavity of the tolerance constraints, the monotone and threshold structure of the optimal policy does not hold. The Whittle Index approach (used, for example, in \cite{zou2021minimizing,hsu2018age}) which have relatively low complexity also does not apply to our multi-channel scheduling problems since each user \ati{in our setting is allowed} to transmit over multiple channels \ati{simultaneously}, whereby the Whittle's Indexability condition does not hold. \ati{Using the generally applicable} LP-based approach \ati{reveals} key insights \ati{that can} guide the designers in developing efficient schedulers for future multi-channel wireless technologies.}

\section{Online Scheduling under Unknown Channel Statistics}
\label{sec:online}

Until this point, we have assumed that the channel success probabilities are known when solving the optimization problems. In this section, we use a Lyapunov-drift-plus-penalty approach(see \cite{neely2010stochastic}) to solve the multi-user online age related optimization problem in the scenario when only the current channel states are known, but the channel statistics are unknown.

We will transfer all the energy and age-related constraints into the virtual queues and view the objective as a penalty term with parameter $M$. For the energy constraint of the source $i$, let us define the corresponding virtual queue as $Q_{1,i}[t]$, whose
initial value is $Q_{1,i}[0]=0$ and  update equation is:
$$Q_{1,i}[t+1]=\left(Q_{1,i}[t]+ \dU_i\left(\bm{\Age}[t]\right)-b_{i}\right)^{+}.$$
Similarly, we define the virtual queue $Q_{2,k}[t]$ for the $k^{th}$ age-related constraint, whose initial value is $Q_{2,k}[0]=0$ and update equation is:
$$Q_{2,k}[t+1]=\left(Q_{2,k}[t]+ \omega_k\left(\bm{\Age}[t]\right)-c_k\right)^{+}.$$
Generically, if the virtual queue $Q_{1,i}[t]$ is stable, then its input rate $\dsp \lim _{T \rightarrow \infty} \frac{1}{T} \sum_{t=1}^{T} \mathbb{E}\left[\dU_i\left(\bm{\Age}[t]\right)\right]$ will be less than its output rate $b_i$ \cite{neely2010stochastic}, so that the corresponding constraint can be satisfied. 
Define the state of both virtual queues and age at time $t$ as $\bm{Q}[t]=(Q_{1,1}[t],\cdots,Q_{1,n}[t],Q_{2,1}[t],\cdots,Q_{2,K}[t], \bm{A}[t]).$ Based on the virtual queues, we will define the quadratic Lyapunov function as:
$$V[t]=\frac{1}{2}(\sum_{i=1}^n Q_{1,i}^2[t]+\sum_{k=1}^K Q_{2,k}^2[t]),$$
and develop an online algorithm to greedily minimize the upper bound of the Lyapunov-drift-plus-penalty function $\Delta V(\mathbf{q})+M\mathbb{E}[\omega_0(\bm{a})]$ given the current state $\bm{q}=(q_{1,1},\cdots,q_{1,n},q_{2,1},\cdots,q_{2,K},\bm{a}),$ where:
$$\Delta V(\mathbf{q})=\mathbb{E}[V[t]-V[t-1]|\bm{Q}[t]=\bm{q}].$$

We consider the multi-user single-channel scheduling problem under tolerance constraints as a specific example to present the design. Since there are no energy constraints, we do not need the set of virtual queues $\{Q_{1,i}[t]\}_i$. In order to express the $k^{th}$ violation rate constraint for source $k=1,\cdots,n$, we let    $\omega_k\left(\bm{\Age}[t]\right)=\mathbb{1}\left(A_{k}[t+1]>\Tr_{k}\right)$ and $c_k=\Tol_k$. 
Then the virtual queue $Q_{2,k}[t]$, whose initial value is $Q_{2,k}[t]=0$, updates as follows:
$$Q_{2,k}[t+1]=\left(Q_{2,k}[t]+ \mathbb{1}\left(A_{k}[t+1]>\Tr_{k}\right)-\Tol_{k}\right)^{+},$$
where $A_{k}[t+1]=1+A_{k}[t](1-S_{k}[t] U_{k}[t])$; $S_{k}[t]$ represents the channel success; $U_{k}[t]$ represents whether the source is scheduled to transmit or not. If virtual queue $Q_{2,k}[t]$ is stable, its input rate,  the threshold violation rate $\spi_k(\Tr_k+1)=\lim_{T \rightarrow \infty} \frac{1}{T} \sum_{t=1}^{T} \mathbb{1}\left(A_{k}[t+1]>\Tr_{k}\right),$ will be less than its output rate $\Tol_{k}.$ 

The conditional Lyapunov drift can be bounded as follows:
\begin{equation*}
    \begin{array}{ll}
         &\Delta V(\mathbf{q})\\
         \leq &\dsp\sum_{k=1}^n q_{2,k} \mathbb{E}\left[R_{k}-\Tol_{k} | q_{2,k}\right]
         + \sum_{k=1}^{n} \mathbb{E}\left[\frac{\left(R_{k}-\Tol_{k}\right)^{2}}{2} | q_{2,k}\right],
    \end{array}
\end{equation*}
where $R_{k} \stackrel{\Delta}=\mathbb{1}\{1+A_{k}\left(1-S_{k} C_{k}\right)>\Tr_k\}.$
At every time slot $t,$ we can develop an online algorithm as summarized below to greedily minimize the upper bound of the Lyapunov drift given the queue lengths $\bm{Q}[t-1]$ and $\bm{A}[t-1]$ since there is no objective or penalty term in this case. 
\begin{algorithm}
	\caption{A Heuristic Scheduling Policy} 
	\begin{algorithmic}[1]
		\State Input current system state: $A_i[t]$,$Q_i[t]$.
		\State Define available transmission decision set: only one $U_i[t]$ can be $1$.
		\State Choose $\bm{U}[t]$ to minimize the upper bound of Lyapunov drift function in the above inequality.
		\State Update queue lengths for next time slot.
	\end{algorithmic} 
\end{algorithm}

Again, for the sake of easy visualization, we will only present the simulation results for the two-user online scheduling problem under age tolerance constraints, but the online algorithm can be simply applied to any number of sources. The simulation results are illustrated in Fig~\ref{fig:U2Region_sym} and Fig~\ref{fig:U2Region_asm} for different parameters where the upper right area of the dash-dot purple line is the stability region of the online scheduler when the channel condition $\cm_i$. The comparison will be in the next section.
\section{Comparison of Stability Regions under age violation constraints}
\label{sec:compare}
In this section, we compare the performance of three different algorithms for the two-user single channel scheduling feasibility problem under age violation tolerance constraints. These are: the optimal scheduler from Section~\ref{sec:multi}; the prior design from \cite{Li2021Soft} developed for a single-channel multi-user setting; and our online scheduler from Section~\ref{sec:online} that does not require channel statistics. 

\begin{figure}[b!]
\begin{center}
\vspace{-0.15in}
\includegraphics[height=5cm]{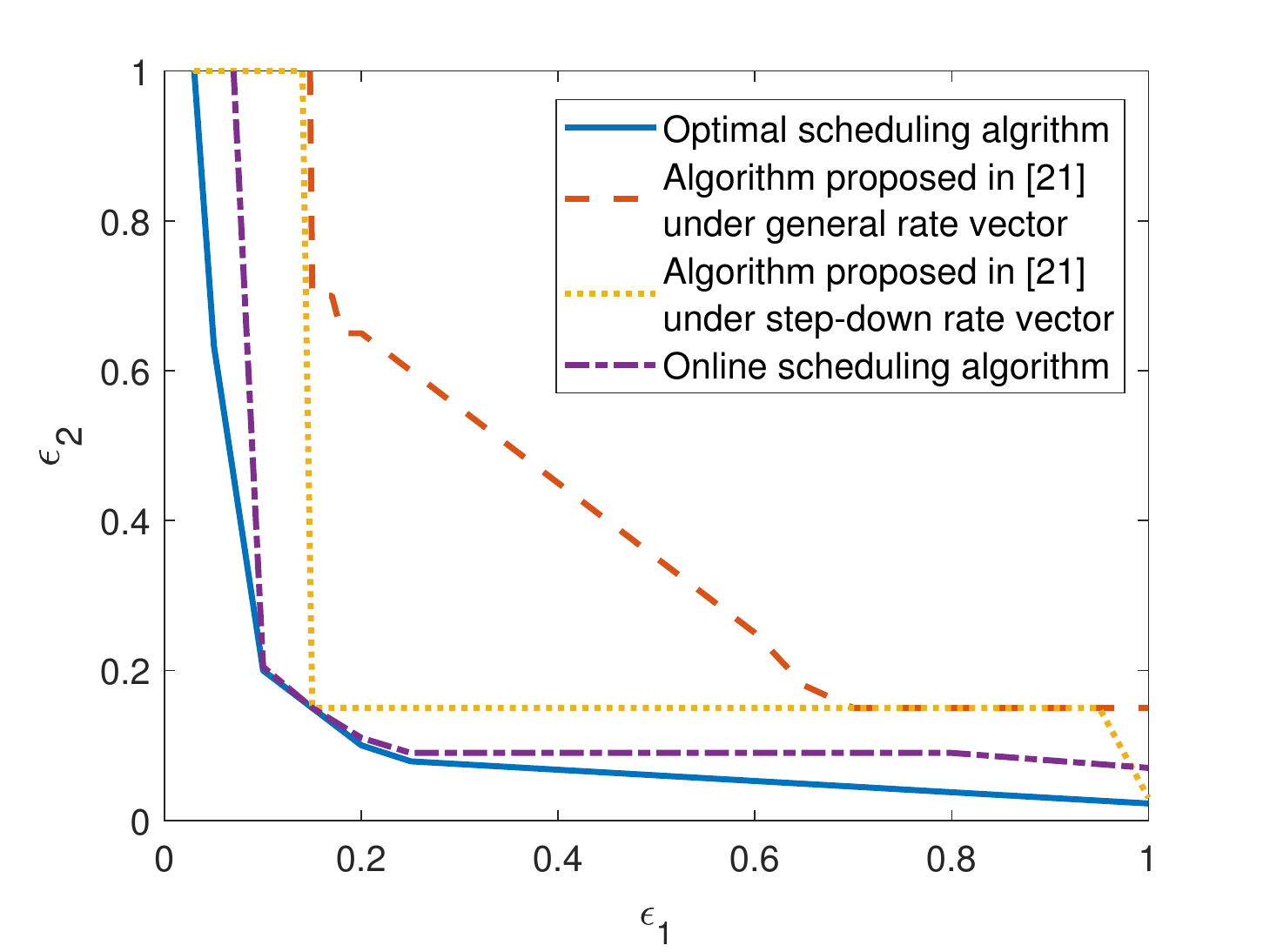}
\vspace{-0.15in}
\caption{Stability region (upper-righter) comparison for symmetric case.}
\vspace{-0.15in}
\label{fig:U2Region_sym}
\end{center}
\end{figure}

We first focus on the case when the two source nodes are symmetric. In Figure~\ref{fig:U2Region_sym}, there are two source nodes with the same age thresholds of $\Tr_1=\Tr_2=2$ and the same channel success probabilities of $\cm_1=\cm_2=0.85$. The upper right area of the blue line is the stability region for the optimal scheduling algorithm in Section~\ref{sec:U2opt}. The yellow and orange lines correspond to the algorithm in \cite{Li2021Soft} and capture the two cases when the rate vector does or does not possess a special property (called step-down rate vector). The purple line marks the stability region for the online algorithm when the channel conditions $\cm_1, \cm_2$ are unknown. 
Several observations are in order from these simulation results: 
\begin{enumerate}[label=(\roman*)]
\item The stability regions are all symmetric, as can be expected due to the homogeneous deadline thresholds and channel conditions.
\item The optimum policy (blue line) outperforms other 
policies, with markedly better performance in cases where the tolerance levels are greatly different from each other. 
\item The online algorithm (purple line) performs very closely to the optimal policy, experiencing a small performance loss only at some extreme range of tolerance levels. 
\item When compared with the algorithms from \cite{Li2021Soft}(yellow and red lines), the online algorithm performs particularly better when one of the tolerance rates is smaller than the corresponding channel loss probability, as observed by the vertical gap between purple and yellow lines. 
\item The online and optimal policies are continuous with respect to the tolerance level, which eliminates the need to check if the tolerance rate vector satisfies certain properties, such as the step-down rate condition in \cite{Li2021Soft}. 
\end{enumerate}

To compare the advantages and disadvantages of the algorithms under non-homogeneous scenarios, in Figure~\ref{fig:U2Region_asm}, we consider two source nodes with asymmetric age thresholds of $\Tr_1=2, \Tr_2=4$ and a common channel success probability of $\cm_1=\cm_2=0.85$. 
Since the violation rate depends on both the age thresholds and the channel success probabilities, this is a non-homogeneous scenario even though $\cm_1=\cm_2.$
In this figure, in contrast to the previous figure, we can further see that the optimal policy outperforms others when one of the tolerance constraints is very strict, namely when $\Tol_1$ approaches $1$. In this regime, the feasible tolerance level $\Tol_2$ of user $2$ other algorithms is bounded away from zero while the optimal algorithm decreases towards zero. 

\begin{figure}[t!]
\begin{center}
\vspace{-0.15in}
\includegraphics[height=5cm]{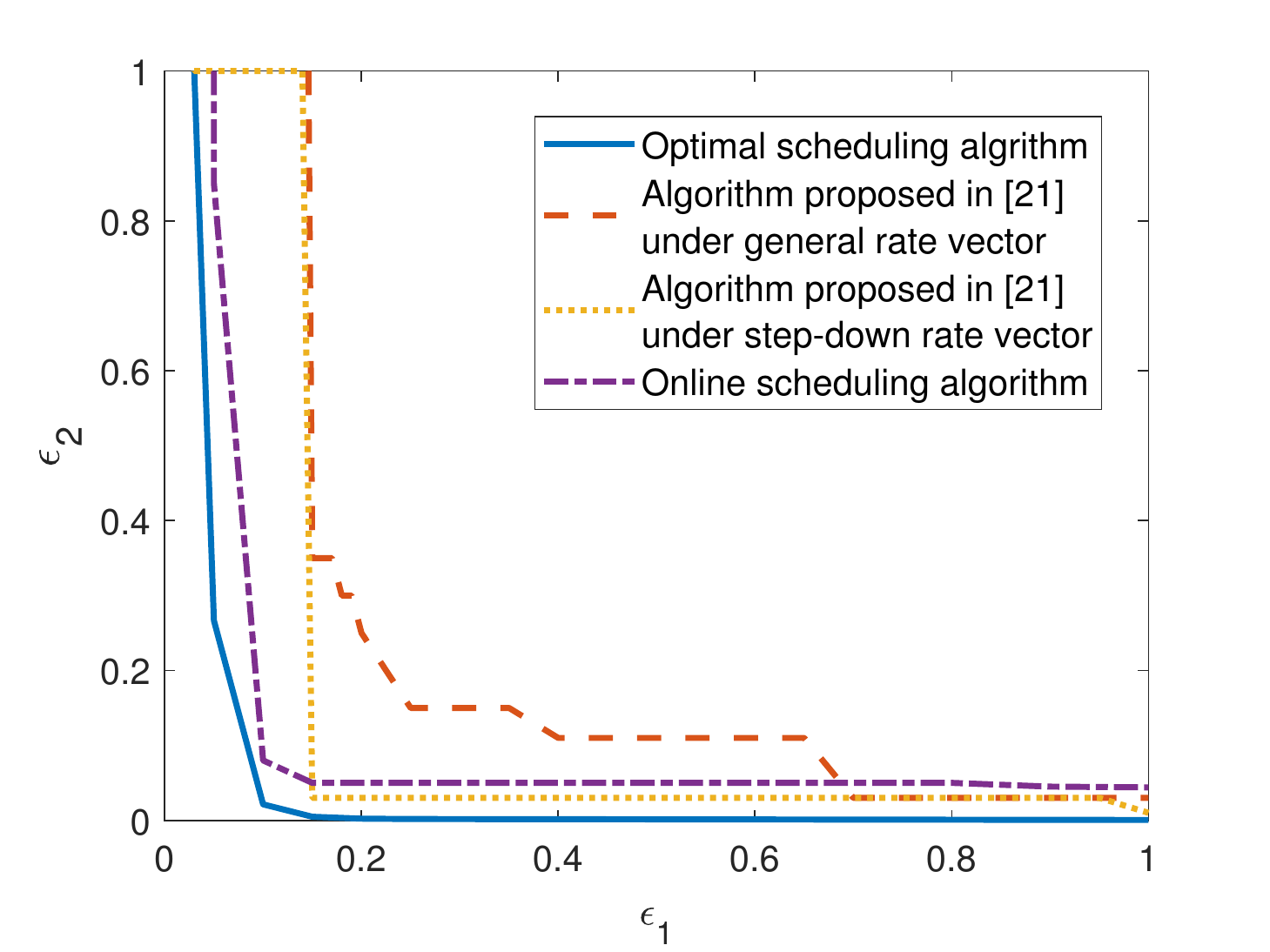}
\vspace{-0.15in}
\caption{Stability region (upper-righter) comparison for asymmetric case.}
\vspace{-0.15in}
\label{fig:U2Region_asm}
\end{center}
\end{figure}

These simulation results are typical of other circumstances, with the common observation that our online scheduler performs close to the optimal scheduler and typically non-negligibly better than the most closely related state-of-art algorithm from \cite{Li2021Soft}, despite the fact that it operates without the knowledge of channel statistics that is assumed in the other designs. 

\section{Conclusions}

In this paper, we considered a general class of age-optimal scheduling problems for multi-source multi-channel communication. We formulated the generic age-optimization problem with flexible weight functions $\omega_k$ under energy and tolerance constraints in the form of a CMDP. We solved this generic problem, {\color{blue}which a usual \ati{threshold-based} structure policy does not apply,} by relating it to the solution an associated linear programming problem using the powerful theory of CMDPs. 
Then, we focused on the special case of single-source multi-channel scenario to investigate the characteristics of optimal scheduler for the important special cases of average-age
and violation-rate minimization. 

Our investigations revealed several interesting insights, including the observation that age-violation-rate minimizing scheduler employs a super-linearly like growing energy allocation strategy with increasing age, as opposed to the sub-linearly like growing allocation for the average-age-minimizing scheduler. These insights may provide useful guidelines for IoT network designers in developing effective update strategies based on different sensitivities of applications to age performance. 

We also studied the special case of multi-source single-channel scheduling problem with age violation rate constraints to investigate the feasibility region of the optimal scheduler together with that of most closely related prior works. Finally, we have developed an online scheduler that does not require the knowledge of channel statistics, and compared its performance to the optimal scheduler through simulations to observe that it performs closely to the optimal scheduler despite its lack of information on channel statistics. 



\newpage 
\bibliographystyle{IEEEtran}
\bibliography{IEEEabrv,main}

\begin{thebibliography}{10}
\providecommand{\url}[1]{#1}
\csname url@samestyle\endcsname
\providecommand{\newblock}{\relax}
\providecommand{\bibinfo}[2]{#2}
\providecommand{\BIBentrySTDinterwordspacing}{\spaceskip=0pt\relax}
\providecommand{\BIBentryALTinterwordstretchfactor}{4}
\providecommand{\BIBentryALTinterwordspacing}{\spaceskip=\fontdimen2\font plus
\BIBentryALTinterwordstretchfactor\fontdimen3\font minus
  \fontdimen4\font\relax}
\providecommand{\BIBforeignlanguage}[2]{{%
\expandafter\ifx\csname l@#1\endcsname\relax
\typeout{** WARNING: IEEEtran.bst: No hyphenation pattern has been}%
\typeout{** loaded for the language `#1'. Using the pattern for}%
\typeout{** the default language instead.}%
\else
\language=\csname l@#1\endcsname
\fi
#2}}
\providecommand{\BIBdecl}{\relax}
\BIBdecl

\bibitem{federal2016fcc}
F.~C. Commission \emph{et~al.}, ``Fcc adopts rules to facilitate next
  generation wireless technologies,'' \emph{FCC, July}, vol.~14, 2016.

\bibitem{kaul2011minimizing}
S.~Kaul, M.~Gruteser, V.~Rai, and J.~Kenney, ``Minimizing age of information in
  vehicular networks,'' in \emph{2011 8th Annual IEEE communications society
  conference on sensor, mesh and ad hoc communications and networks}.\hskip 1em
  plus 0.5em minus 0.4em\relax IEEE, 2011, pp. 350--358.

\bibitem{kaul2012real}
S.~Kaul, R.~Yates, and M.~Gruteser, ``Real-time status: How often should one
  update?'' in \emph{2012 Proceedings IEEE INFOCOM}.\hskip 1em plus 0.5em minus
  0.4em\relax IEEE, 2012.

\bibitem{kosta2017age}
A.~Kosta, N.~Pappas, V.~Angelakis \emph{et~al.}, ``Age of information: A new
  concept, metric, and tool,'' \emph{Foundations and Trends{\textregistered} in
  Networking}, vol.~12, no.~3, pp. 162--259, 2017.

\bibitem{chen2020age}
H.~Chen, Y.~Gu, and S.-C. Liew, ``Age-of-information dependent random access
  for massive iot networks,'' in \emph{INFOCOM WKSHPS}.\hskip 1em plus 0.5em
  minus 0.4em\relax IEEE, 2020, pp. 930--935.

\bibitem{Zhou2022Efficient}
X.~Zhou, I.~Koprulu, A.~Eryilmaz, and M.~J. Neely, ``Efficient distributed mac
  for dynamic demands: Congestion and age based designs,'' \emph{IEEE/ACM
  Transactions on Networking}, pp. 1--14, 2022.

\bibitem{abolhassani2022fresh}
B.~Abolhassani, J.~Tadrous, A.~Eryilmaz, and E.~Yeh, ``Fresh caching of dynamic
  content over the wireless edge,'' \emph{IEEE/ACM Transactions on Networking},
  2022.

\bibitem{liu2019proactive}
R.~Liu, E.~Yeh, and A.~Eryilmaz, ``Proactive caching for low access-delay
  services under uncertain predictions,'' \emph{Proceedings of the ACM on
  Measurement and Analysis of Computing Systems}, vol.~3, no.~1, 2019.

\bibitem{tang2019scheduling}
H.~Tang, J.~Wang, L.~Song, and J.~Song, ``Scheduling to minimize age of
  information in multi-state time-varying networks with power constraints,'' in
  \emph{2019 57th Annual Allerton Conference on Communication, Control, and
  Computing (Allerton)}.\hskip 1em plus 0.5em minus 0.4em\relax IEEE, 2019, pp.
  1198--1205.

\bibitem{sombabu2018age}
B.~Sombabu and S.~Moharir, ``Age-of-information aware scheduling for
  heterogeneous sources,'' in \emph{Proceedings of the 24th Annual
  International Conference on Mobile Computing and Networking}, 2018, pp.
  696--698.

\bibitem{li2020age}
M.~Li, C.~Chen, H.~Wu, X.~Guan, and X.~Shen, ``Age-of-information aware
  scheduling for edge-assisted industrial wireless networks,'' \emph{IEEE
  Transactions on Industrial Informatics}, vol.~17, no.~8, pp. 5562--5571,
  2020.

\bibitem{jhunjhunwala2018age}
P.~R. Jhunjhunwala and S.~Moharir, ``Age-of-information aware scheduling,'' in
  \emph{SPCOM}.\hskip 1em plus 0.5em minus 0.4em\relax IEEE, 2018.

\bibitem{han2020joint}
D.~Han, W.~Chen, and Y.~Fang, ``Joint channel and queue aware scheduling for
  latency sensitive mobile edge computing with power constraints,'' \emph{IEEE
  Transactions on Wireless Communications}, vol.~19, no.~6, pp. 3938--3951,
  2020.

\bibitem{hu2021status}
L.~Hu, Z.~Chen, Y.~Dong, Y.~Jia, L.~Liang, and M.~Wang, ``Status update in iot
  networks: Age-of-information violation probability and optimal update rate,''
  \emph{IEEE Internet of Things Journal}, vol.~8, no.~14, 2021.

\bibitem{pappas2019delay}
N.~Pappas and M.~Kountouris, ``Delay violation probability and age of
  information interplay in the two-user multiple access channel,'' in
  \emph{20th International Workshop on SPAWC}.\hskip 1em plus 0.5em minus
  0.4em\relax IEEE, 2019, pp. 1--5.

\bibitem{chen2020minimizing}
M.~Chen, Y.~Xiao, Q.~Li, and K.-c. Chen, ``Minimizing age-of-information for
  fog computing-supported vehicular networks with deep q-learning,'' in
  \emph{ICC 2020-2020 IEEE International Conference on Communications
  (ICC)}.\hskip 1em plus 0.5em minus 0.4em\relax IEEE, 2020, pp. 1--6.

\bibitem{sun2019sampling}
Y.~Sun and B.~Cyr, ``Sampling for data freshness optimization: Non-linear age
  functions,'' \emph{Journal of Communications and Networks}, vol.~21, no.~3,
  pp. 204--219, 2019.

\bibitem{emara2020spatiotemporal}
M.~Emara, H.~Elsawy, and G.~Bauch, ``A spatiotemporal model for peak aoi in
  uplink iot networks: Time versus event-triggered traffic,'' \emph{IEEE
  internet of things journal}, vol.~7, no.~8, pp. 6762--6777, 2020.

\bibitem{li2017emulating}
B.~Li, A.~Eryilmaz, and R.~Srikant, ``Emulating round-robin in wireless
  networks,'' in \emph{Proceedings of the 18th ACM International Symposium on
  Mobile Ad Hoc Networking and Computing}, 2017, pp. 1--10.

\bibitem{dong2019age}
Y.~Dong, Z.~Chen, S.~Liu, P.~Fan, and K.~B. Letaief, ``Age-upon-decisions
  minimizing scheduling in internet of things: To be random or to be
  deterministic?'' \emph{IEEE Internet of Things Journal}, vol.~7, no.~2, 2019.

\bibitem{Li2021Soft}
C.~Li, Q.~Liu, S.~Li, Y.~Chen, Y.~T. Hou, and W.~Lou, ``On scheduling with aoi
  violation tolerance,'' in \emph{IEEE INFOCOM}, 2021, pp. 1--9.

\bibitem{song2021analysis}
M.~Song, H.~H. Yang, H.~Shan, J.~Lee, H.~Lin, and T.~Q. Quek, ``Analysis of aoi
  violation probability in wireless networks,'' in \emph{17th ISWCS}.\hskip 1em
  plus 0.5em minus 0.4em\relax IEEE, 2021.

\bibitem{devassy2018delay}
R.~Devassy, G.~Durisi, G.~C. Ferrante, O.~Simeone, and E.~Uysal-Biyikoglu,
  ``Delay and peak-age violation probability in short-packet transmissions,''
  in \emph{ISIT}.\hskip 1em plus 0.5em minus 0.4em\relax IEEE, 2018, pp.
  2471--2475.

\bibitem{hu2020optimal}
L.~Hu, Z.~Chen, Y.~Dong, Y.~Jia, M.~Wang, L.~Liang, and C.~Chen, ``Optimal
  status update in iot systems: An age of information violation probability
  perspective,'' in \emph{VTC2020-Fall}.\hskip 1em plus 0.5em minus 0.4em\relax
  IEEE, 2020, pp. 1--5.

\bibitem{altman1999constrained}
E.~Altman, \emph{Constrained Markov decision processes: stochastic
  modeling}.\hskip 1em plus 0.5em minus 0.4em\relax Routledge, 1999.

\bibitem{dolgov2005stationary}
D.~A. Dolgov and E.~H. Durfee, ``Stationary deterministic policies for
  constrained mdps with multiple rewards, costs, and discount factors,'' in
  \emph{IJCAI}, vol.~19.\hskip 1em plus 0.5em minus 0.4em\relax Citeseer, 2005,
  pp. 1326--1331.

\bibitem{satija2020constrained}
H.~Satija, P.~Amortila, and J.~Pineau, ``Constrained markov decision processes
  via backward value functions,'' in \emph{International Conference on Machine
  Learning}.\hskip 1em plus 0.5em minus 0.4em\relax PMLR, 2020, pp. 8502--8511.

\bibitem{neely2010stochastic}
M.~J. Neely, ``Stochastic network optimization with application to
  communication and queueing systems,'' \emph{Synthesis Lectures on
  Communication Networks}, vol.~3, no.~1, pp. 1--211, 2010.

\bibitem{hsu2018age}
Y.-P. Hsu, ``Age of information: Whittle index for scheduling stochastic
  arrivals,'' in \emph{ISIT}.\hskip 1em plus 0.5em minus 0.4em\relax IEEE,
  2018, pp. 2634--2638.

\bibitem{tripathi2017age}
V.~Tripathi and S.~Moharir, ``Age of information in multi-source systems,'' in
  \emph{GLOBECOM}.\hskip 1em plus 0.5em minus 0.4em\relax IEEE, 2017, pp. 1--6.

\bibitem{zou2021minimizing}
Y.~Zou, K.~T. Kim, X.~Lin, and M.~Chiang, ``Minimizing age-of-information in
  heterogeneous multi-channel systems: A new partial-index approach,'' in
  \emph{Proceedings of the Twenty-second International Symposium on Theory,
  Algorithmic Foundations, and Protocol Design for Mobile Networks and Mobile
  Computing}, 2021, pp. 11--20.

\bibitem{elgabli2019reinforcement}
A.~Elgabli, H.~Khan, M.~Krouka, and M.~Bennis, ``Reinforcement learning based
  scheduling algorithm for optimizing age of information in ultra reliable low
  latency networks,'' in \emph{ISCC}.\hskip 1em plus 0.5em minus 0.4em\relax
  IEEE, 2019, pp. 1--6.

\bibitem{li2020learning}
M.~Li, C.~Chen, C.~Hua, and X.~Guan, ``Learning-based autonomous scheduling for
  aoi-aware industrial wireless networks,'' \emph{IEEE Internet of Things
  Journal}, vol.~7, no.~9, pp. 9175--9188, 2020.

\bibitem{tang2020minimizing}
H.~Tang, J.~Wang, L.~Song, and J.~Song, ``Minimizing age of information with
  power constraints: Multi-user opportunistic scheduling in multi-state
  time-varying channels,'' \emph{IEEE Journal on Selected Areas in
  Communications}, vol.~38, no.~5, pp. 854--868, 2020.

\end{thebibliography}

\end{document}